%% file: paper.tex
\documentclass [11pt]{article}

\usepackage{amsmath,amsfonts,amsthm,amssymb,fullpage}
\usepackage{graphics,color,url}
\usepackage{graphicx}
\usepackage{epsfig}

\newtheorem{claim}{Claim}[section]

\newtheorem{defi}{Definition}[section]

\newtheorem{lemma}{Lemma}[section]

\newtheorem{theorem}{Theorem}[section]

\newcommand{\calM}{{\cal M}}

\newcommand{\calR}{{\cal R}}

\newcommand{\Gmodular}{\textup{{\sc Greedy-SM}}}
\newcommand{\Mmodular}{\textup{{\sc Mechanism-SM}}}
\newcommand{\RMmodular}{\textup{{\sc Random-Mechanism-SM}}}

\newcommand{\gre}{\textup{{\sc Greedy-KS}}}
\newcommand{\Mknapsack}{\textup{{\sc Mechanism-KS}}}
\newcommand{\RMknapsack}{\textup{{\sc Random-Mechanism-KS}}}

\newcommand{\fHK}{\textup{{\sc Fraction-HK}}}
\newcommand{\greH}{\textup{{\sc Greedy-HK}}}
\newcommand{\MHknapsack}{\textup{{\sc Mechanism-HK}}}
\newcommand{\RMHknapsack}{\textup{{\sc Random-Mechanism-HK}}}

\title{On the Approximability of Budget Feasible Mechanisms}
\author{\\Ning Chen\footnote{Division of Mathematical Sciences, School of Physical and Mathematical Sciences, Nanyang Technological University, Singapore. Email: {\tt ningc@ntu.edu.sg, ngravin@pmail.ntu.edu.sg}. Work done while visiting Microsoft Research Asia.} \and \\Nick Gravin\footnotemark[1]
\footnote{St. Petersburg Department of Steklov Mathematical Institute RAS, Russia.}
\and \\Pinyan Lu\footnote{Microsoft Research Asia. Email: {\tt pinyanl@microsoft.com}.}}
\date{}

\begin{document}

\maketitle
\begin{abstract}
Budget feasible mechanisms, recently initiated by Singer (FOCS 2010), extend algorithmic mechanism design problems to a realistic setting with a budget constraint.
We consider the problem of designing truthful budget feasible mechanisms for general submodular functions: we give a randomized mechanism with approximation ratio $7.91$ (improving the previous best-known result 112), and a deterministic mechanism with approximation ratio $8.34$.
Further we study the knapsack problem, which is special submodular function, give a $2+\sqrt{2}$ approximation deterministic mechanism (improving the previous best-known result 6), and a 3 approximation randomized mechanism. We provide a similar result for an extended knapsack problem with heterogeneous items, where items are divided into groups and one can pick at most one item from each group.

Finally we show a lower bound of approximation ratio of $1+\sqrt{2}$ for deterministic mechanisms and 2 for randomized mechanisms for knapsack, as well as the general submodular functions. Our lower bounds are unconditional, which do not rely on any computational or complexity assumptions.
\end{abstract}

\thispagestyle{empty}
\newpage
\setcounter{page}{1}

\section{Introduction}

It is well-known that a mechanism may have to pay a large amount to enforce incentive compatibility (i.e., truthfulness). For example, the seminal VCG mechanism may have unbounded payment (compared to the shortest path) in path auctions~\cite{AT02}.
The negative effect of truthfulness on payments leads to a broad study of frugal mechanism design, i.e., how should one minimize his payment to get a desired output with incentive agents? While a class of results have been established~\cite{AT02,talwar03,KKT05,KSM10,CEG10}, in practice, one cannot expect a negative overhead for a few perspectives, e.g., budget or resource limit.

Recently, Singer~\cite{PS10} consider mechanism design problems from a converse angle and initiate the study of truthful mechanism design with a sharp budget constraint: the total payment of a mechanism is upper bounded by a given value $B$. Formally, in a marketplace each agent/item has a {\em privately} known incurred cost $c_i$. For any given subset $S$ of agents, there is a {\em publicly} known valuation $v(S)$, meaning the social welfare derived from $S$. A mechanism selects a subset $S$ of agents and decides a payment $p_i$ to each $i\in S$. Agents bid strategically on their costs and would like to maximize their utility $p_i-c_i$. The objective is to design truthful budget feasible mechanisms with outputs approximately close to socially optimal solution. In other words, it studies the ``price of being truthful" in a budget constraint framework\footnote{Note that if we do not consider truthful mechanism design, the problem is purely an optimization question with an extra capacity (i.e., budget) constraint, which has been well-studied in, e.g.,~\cite{NW81,Svi04,KG05,FMV07,LMN09}, in the framework of submodularity. It is well-known that a simple greedy algorithm gives the best possible approximation ratio $1-1/e$~\cite{NW81}. When agents are weighted (corresponding to costs in our setting), the simple greedy algorithm may have unbounded approximation ratio~\cite{KMN99}; a variant of the greedy by picking the maximum of the original greedy and the agent with the largest value provides the same approximation ratio $1-1/e$~\cite{KG05}.}.

Although budget is a realistic condition that appears almost everywhere in daily life, it has not received much attention until very recently~\cite{DLN08,BGG10,CHM10,PS10}. In the framework of worst case analysis, most results are negative~\cite{DLN08}. The introduction of budget adds another dimension to mechanism design; it further limits the searching space, especially given the (already) strong restriction of truthfulness. Designing budget feasible mechanisms even requires us to bound the threshold payment of each individual, which, not surprisingly, is tricky to analyze.

While the problem in general does not admit any budget feasible mechanism\footnote{For example, one with budget $B=1$ would like to purchase a path from $s$ to $t$ in a network $\{(s,v),(v,t)\}$ where each edge has incurred cost 0. In any truthful mechanism that guarantees to buy the path (i.e., outputs the socially optimum solution), one has to pay each edge the threshold value $B$, leading to a total payment $2B$ which exceeds the given budget.}, Singer~\cite{PS10} studied an important class of valuation functions, i.e., monotone submodular functions. He gives a randomized truthful mechanism with constant approximation ratio 112 for any monotone submodular functions, and deterministic mechanisms for special cases including knapsack (ratio 6) and coverage. Further, he shows that no deterministic truthful mechanism can obtain an approximation ratio better than 2 even for knapsack.

\subsection{Our Results}

In this paper, we improve upper and lower bounds of budget feasible mechanisms for general submodular functions and knapsack, summarized in the following table.

\small
\begin{table*}[t]
\begin{center}
\begin{tabular}{|c|c|c|c|c|c|c|c|c|}\hline
 & \multicolumn{4}{c|}{Submodular functions} & \multicolumn{4}{c|}{Knapsack} \\
 \cline{2-9}
 & \multicolumn{2}{c|}{deterministic} & \multicolumn{2}{c|}{randomized} & \multicolumn{2}{c|}{deterministic} & \multicolumn{2}{c|}{randomized}  \\
 \cline{2-9}
 & upper & lower & upper & lower & upper & lower & upper & lower  \\ \hline
Singer~\cite{PS10} & $-$ & 2 & 112 & $-$ & $6$ & $2$ & $-$ & $-$ \\ \hline
Our results & $8.34^*$ & $1+\sqrt{2}$ & 7.91 & 2 & $2+\sqrt{2}$ & $1+\sqrt{2}$ & 3 & 2 \\
\hline
\end{tabular}
\\[.1in] *It may require exponential running time for general submodular functions.
%\caption{*It requires exponential running time for general submodular functions.}
\end{center}
\end{table*}
\normalsize

In truthful mechanism design, if there is no restriction on total payment, it is sufficient to focus on designing monotone allocations --- the payment to each individual winner is the unique threshold to maintain the winning status~\cite{myerson}. With a sharp budget constraint, in addition to monotone allocation, we also have to upper bound the sum of threshold payments. For submodular functions, the natural greedy algorithm is a good candidate for designing budget feasible mechanisms due to its nice monotonicity and small approximation ratio. However, the threshold payment to each winner can be very complicated because an agent can bid differently to get other ranking positions in the greedy algorithm, which results in different ways of computing the marginal contributions for the rest agents, and therefore, unpredictably change the set of winners.
Singer~\cite{PS10} bound the threshold of each winner by considering all possible ranking positions for his bids and taking the maximum of the thresholds of all these positions. In Section~3, we give a clean and tight analysis for the upper bound on threshold payment by applying the combinatorial structure of submodular functions (Lemma~\ref{lemma-pay-upper}). These upper bounds on payments suggest appropriate parameters in our randomized mechanism, which roughly speaking, selects the greedy algorithm or the agent with the largest value at a certain probability.

A difficulty of deriving deterministic mechanisms is related to the agent $i^*$ with the largest value $v(i^*)$ --- due to its (possibly large) cost greedy may not include it, which could result in a solution with an arbitrarily bad ratio. However, we cannot simply compare greedy with $v(i^*)$ because this breaks monotonicity as $i^*$ is able to manage the greedy solution by his bid (this is exactly where randomization helps). To get around of this issue, we drop $i^*$ out of the market and compare $v(i^*)$ with remaining agents in an appropriate way --- now $i^*$ is completely independent of the rest of the market and cannot affect its output --- this gives our deterministic mechanisms for submodular functions and knapsack with small approximation ratios (note that we still need to be careful about the agents in the remaining market as they are still able to manage their bids to beat $v(i^*)$).

On the other hand, it is interesting to explore limitations of budget feasible mechanisms. Singer gives a simple lower bound of $2$ on the approximation ratio and proposes that exploring the lower bounds that are dictated by budget feasibility is ``perhaps the most interesting question"~\cite{PS10}. In Section~4, we prove a stronger lower bound $1+\sqrt{2}$ for deterministic mechanisms. In most lower bounds proofs for truthful mechanisms, a number of related instances are constructed and one shows that a truthful mechanism cannot do well for all of them~\cite{CKV07,KV07,NR01,PT09}. (For example, in Singer's proof, three instances are constructed.) Our lower bound proof uses a slightly different approach: We first establish a property of a truthful mechanism for all instances provided that the mechanism has a good approximation ratio
(Lemma~\ref{lemma-knapsack-lower}), then we conclude that this property is inconsistent with the budget feasibility condition for a carefully constructed instance. Furthermore, we show a lower bound of $2$ for universally randomized budget feasible mechanisms. Both our lower bounds are independent of computational assumptions and hold for instances with small number of agents.

While submodular functions admit a good approximation budget feasible mechanism, its generalization seems to be a very difficult task and there may not be any good approximation mechanisms for instances like path and spanning tree~\cite{PS10}. In Section~5, we take a first step of this generalization by considering an extended knapsack problem with heterogeneous items, where items are of different types and we are only allowed to pick one item from each type.
Here we cannot apply the same greedy mechanism for the original knapsack as it may not even generate a feasible solution; and its approximation ratio can be arbitrarily bad if we only take the first agent from each type.
%One may argue that heterogeneous knapsack, in fact, can be treated as submodular function, if we extend the valuation function to any set $S$ by taking maximal valuation over all feasible subsets of $S$. But we cannot use the same greedy as it may not even generate a feasible solution.
To construct a truthful mechanism with good approximation, we employ a greedy strategy with {\em deletions} --- in the process of the greedy, we either add a new item whose type has not been considered yet or replace an existing item with the new one of same type.
Although there are deletions, the greedy algorithm is still monotone (but its proof is much more involved), based on which we have similar approximation mechanisms for heterogeneous knapsack. We believe that the greedy strategy with deletions can be extended to a number of interesting non-submodular settings to derive budget feasible mechanisms.

\section{Preliminaries}

In a marketplace there are $n$ agents (or items), denoted by $A=\{1,\ldots,n\}$. Each agent $i$ has a privately known incurred {\em cost} $c_i$ (or denoted by $c(i)$). For any given subset $S\subseteq A$ of agents, there is a publicly known valuation $v(S)$, meaning the social welfare derived from $S$. We assume $v(\emptyset)=0$ and $v(S)\le v(T)$ for any $S\subset T\subseteq A$. We say the valuation function {\em submodular} if $v(S)+v(T)\ge v(S\cap T)+ v(S\cup T)$ for any $S,T\subseteq A$.

Upon receiving a {\em bid} cost $b_i$ from each agent, a mechanism decides an {\em allocation} $S\subseteq A$ as winners and a {\em payment} $p_i$ to each $i\in A$. We assume that the mechanism has no positive transfer (i.e., $p_i=0$ if $i\notin S$) and individually rational (i.e., $p_i\ge b_i$ if $i\in S$). Agents bid strategically on their costs and would like to maximize their utilities, which is $p_i-c_i$ if $i$ is a winner and $0$ otherwise. We say a mechanism is {\em truthful} if it is of best interests for all agents to bid their true costs. For randomized mechanisms, we consider universal truthfulness in this paper (i.e., a randomized mechanism takes a distribution over deterministic truthful mechanisms).

Since our setting is in single parameter domain as each agent has one private cost, it is well-known~\cite{myerson} that a mechanism is truthful if and only if its allocation rule is monotone (i.e., a winner keeps winning if he unilaterally decreases his bid) and the payment to each winner is his threshold bid to win. Therefore, we will only focus on designing monotone allocations and do not specify the payment to each winner explicitly.

A mechanism is said to be {\em budget feasible} if $\sum_i p_i \le B$, where $B$ is a given sharp budget constraint.
Assume without loss of generality that $c_i\le B$ for any agent $i\in A$, since otherwise he will never win in any (randomized) budget feasible truthful mechanism. Our objective is to design truthful budget feasible mechanisms with outputs approximately close to the social optimum. That is, we want to minimize the {\em approximation ratio} of a mechanism, which is defined as $\max_I \frac{\mathcal{M}(I)}{opt(I)}$, where $\mathcal{M}(I)$ is the (expected) value of mechanism $\mathcal{M}$ on instance $I$ and $opt(I)$ is the optimal value of the integer program: $\max_{_{S\subseteq A}}v(S)$ subjected to $c(S)\le B$.

\section{Budget Feasible Mechanisms}

For any given submodular function, denote the marginal contribution of an item $i$ with respect to set $S$ by $m_S(i)=v(S\cup\{i\})-v(S)$.
We assume that agents are sorted according to their increasing marginal
contributions relative to cost, recursively defined by: $i + 1 = \arg\max_{j \in [n]} \frac{m_{S_i}(j)}{c_j}$, where $S_i =\{1,\ldots,i\}$ and $S_0=\emptyset$. To simplify notations we will denote this order by $[n]$ and write $m_i$ instead of $m_{S_{i-1}}(i)$. This sorting, in the presence of submodularity, implies that
\[ \frac{m_1}{c_1} \geq \frac{m_2}{c_2} \geq \cdots \geq \frac{m_n}{c_n}.\]
Notice that $v(S_k) = \sum_{i\leq k} m_i$ for all $k\in [n]$.

The following greedy scheme is the core of our mechanism (where the parameters denote the set of agents $A$ and available budget $B/2$).
\begin{center}
\small{}\tt{} \fbox{
\parbox{3.0in}{
\hspace{0.05in} \\[-0.05in] $\Gmodular(A,B/2)$
\begin{enumerate}
\item Let $k=1$ and $S=\emptyset$
\item While $k\le |A|$ and $c_k\le \frac{B}{2} \cdot \frac{m_k}{\sum_{i\in S\cup\{k\}}m_i}$
\begin{itemize}
\item $S \leftarrow S\cup \{k\}$
\item $k\leftarrow k+1$
\end{itemize}
\item Return winning set $S$
\end{enumerate}
}}
\end{center}
Our mechanism for general submodular functions is as follows.
\begin{center}
\small{}\tt{} \fbox{
\parbox{4.0in}{
\hspace{0.05in} \\[-0.05in] \RMmodular
\begin{enumerate}
\item Let $A=\{i~|~c_i\le B\}$ and $i^*\in \arg\max_{i\in A} v(i)$
\item with probability $\frac{2}{5}$, return $i^*$
\item with probability $\frac{3}{5}$, return $\Gmodular(A,B/2)$
\end{enumerate}
}}
\end{center}

In the above mechanism, if it returns $i^*$, the payment to $i^*$ is $B$; if it returns $\Gmodular(A,B/2)$, the payment is more complicated and is given in~\cite{PS10}. Actually,
we do not need this explicit payment formula to prove our result.

\begin{theorem}\label{theorem-mechanism-SM}
\RMmodular \ is a budget feasible universally truthful mechanism for submodular valuation function with approximation ratio $\frac{5e}{e-1} (\approx 7.91)$.
\end{theorem}

\subsection{Analysis of \RMmodular}

In this subsection we analyze \RMmodular\ in terms of three respects: truthfulness, budget feasibility and approximation. They together yield a proof for Theorem~\ref{theorem-mechanism-SM}.

\subsubsection{Universal Truthfulness}
Our mechanism is a simple random combination of two mechanisms. To prove that the \RMmodular \ is universally truthful, it suffices to prove that these two mechanisms are truthful, i.e., the allocation rule is monotone.

The scheme where we simply return $i^*$ is obviously truthful. Also it is easy to see that prior step when we throw away the agents having the cost greater than $B$ does not affect truthfulness. The greedy scheme $\Gmodular(A,B/2)$ is monotone as well, since any item out of a winning set can not increase its bid and become a winner.

\subsubsection{Budget Feasibility}

While truthfulness is quite straightforward, the budget feasibility part turns out to be quite tricky.
The difficulties arise when we compute the payment to each item. Indeed, it can happen that an item changes its bid (while still remaining in the winning set) to force the mechanism to change its output. In other words, an item can control the output of the mechanism. Fortunately, in such a case no item can reduce the valuation of the output too much. That enables us to write an upper bound on the bid of each item in case of submodularity;  summing up these bounds yields budget feasibility.

If the mechanism returns $i^*$, his payment is $B$ and it is clearly budget feasible. It remains to prove budget feasibility for $\Gmodular(A,B/2)$. A similar but weaker result has been proven in~\cite{PS10} using the characterization of payments and arguing that the total payment is not larger than $B$. Here we directly show that the payment to any item $i$ in the winning set $S$ is bounded above by $\frac{m_i}{v(S)}\cdot B$; then the total payment will be bounded by $B$. Before doing that, we first prove a useful lemma.

\begin{lemma}
\label{submodular average}
Let $S\subset T\subseteq [n]$ and $t_0=\arg\max_{t\in T\setminus S}\frac{m_S(t)}{c(t)}$. Then $$\frac{v(T)-v(S)}{c(T)-c(S)}\le\frac{m_S(t_0)}{c(t_0)}.$$
\end{lemma}

\begin{proof}
Assume for contradiction that the lemma does not hold, then for all $t\in T\setminus S$, we have
$\frac{v(T)-v(S)}{c(T)-c(S)}> \frac{m_S(t)}{c(t)}.$
Then add all inequalities together, we have
$$\frac{v(T)-v(S)}{c(T)-c(S)}> \frac{\sum_{t\in T\setminus S} m_S(t)}{\sum_{t\in T\setminus S} c(t)} = \frac{\sum_{t\in T\setminus S} m_S(t)}{c(T)-c(S)}.$$
This implies that $v(T)-v(S) >\sum_{t\in T\setminus S} m_S(t)$, which contradicts the submodularity.
\end{proof}

Let $1,\ldots, k$ be the order of items in which we add them to the winning set.
Let $\emptyset=S_0\subset S_1\subset\ldots\subset S_k\subseteq [n]$ be the sequence of winning sets that we pick at each step by applying our mechanism. Thus we have $S_j=[j]$.
Now, since $v$ is sumbodular, we can write the following chain of inequalities (note that marginal contribution is smaller for larger sets).
$$\frac{m_{S_0}(1)}{c_1}\ge\frac{m_{S_1}(2)}{c_2}\ge\ldots\ge\frac{m_{S_{k-1}}(k)}{c_k}\ge\frac{2 v(S_k)}{B}.$$

%According to the characterization of one parameter truthful mechanisms payment for any item is defined as a threshold value for this item to be a winner.
%
%To provide the payments be budget feasible it suffices to show that $j\in[k]$  can not bid more than $m_{S_{j-1}}(j)\frac{B}{v(S_k)}$ and still be a winner.
%Indeed $\sum_{j=1}^{k}m_{S_{j-1}}(j)=\sum_{j=1}^{k}v(S_j)-v(S_{j-1})=v(S_k)-v(S_0)=v(S_k)$ and thus $\sum_{j=1}^{k}p_j\le\sum_{j=1}^{k}m_{S_{j-1}}(j)\frac{B}{v(S_k)}=B.$

The following is our main lemma.

\begin{lemma}\label{lemma-pay-upper}
There is no $j\in S=\Gmodular(A,B/2)$ such that it can bid more than $m_{S_{j-1}}(j)\frac{B}{v(S_k)}$ and still get into a winning set. Thus the payment to $j$ is upper bounded by $m_{S_{j-1}}(j)\frac{B}{v(S_k)}$.
\end{lemma}
\begin{proof}
Assume that there is $j\in[k]$ such that it can bid $b_j>m_{S_{j-1}}(j)\frac{B}{v(S_k)}$ and still wins. We will use notation $b$ instead of $c$ to emphasize that we consider a new scenario where $j$ has increased its bid to $b_j$ and others remain the same.

Note that $\frac{m_{S_0}(1)}{c_1}\ge\frac{m_{S_1}(2)}{c_2}\ge\ldots\ge\frac{m_{S_{j-1}}(j)}{c_j}\ge\frac{m_{S_{j-1}}(j)}{b_j}.$ Thus $S_{j-1}$ still get into the winning set.

For bid vector $b$, denote by $S$ the set we have chosen right before $j$ is included into the winning set. Thus we have
\begin{eqnarray}
j &=& \arg\max_{i\in[n]\setminus S}\frac{m_S(i)}{b_i}, \\
\frac{m_S(j)}{b_j} &\ge& \frac{2 v(S\cup\{j\})}{B}.
\end{eqnarray}
We may assume that $S_k\cup S\supsetneq S\cup\{j\}$. Indeed, otherwise $S\cup\{j\}=S_k\cup S$ and $$\frac{m_{S_{j-1}}(j)}{b_j}\ge\frac{m_{S}(j)}{b_j}\ge\frac{2 v(S\cup\{j\})}{B}\ge \frac{2 v(S_k)}{B}\ge\frac{v(S_k)}{B}.$$
Thus $b_j\le m_{S_{j-1}}\frac{B}{v(S_k)}$ and we get a contradiction.

Let $R=S_k\setminus S$. Applying equation~(1) and Lemma~\ref{submodular average} to $S_k\cup S$ and $S\cup \{j\}$, we know that for some $r_0\in R\setminus\{j\}$,
$$\frac{v(S_k\cup S)-v(S\cup\{j\})}{b(S_k\cup S)-b(S\cup\{j\})}\le\frac{m_{S\cup\{j\}}(r_0)}{b(r_0)}\le\frac{m_S(j)}{b_j}.$$

On the other hand we know that $b_j> m_{S_{j-1}}(j)\frac{B}{v(S_k)}$. Hence,
$\frac{m_S(j)}{b_j}< \frac{m_S(j)}{m_{S_{j-1}}(j)}\frac{v(S_k)}{B}<\frac{v(S_k)}{B}.$
Combining these inequalities, we get $$\frac{v(S_k\cup S)-v(S\cup\{j\})}{b(S_k\cup S)-b(S\cup\{j\})}<\frac{v(S_k)}{B}.$$
We have $b(S_k\cup S)-b(S\cup\{j\})=b(R\setminus\{j\})=c(R\setminus\{j\})\le c(S_k).$

Recall that $\frac{m_{S_{i-1}}(i)}{c_i}\ge \frac{2 \emph{v}(S_k)}{B}$ for $i\in[k]$.  Thus $c_i\le m_{S_{i-1}}(i)\frac{B}{2 v(S_k)}$ and
$c(S_k)=\sum_{i=1}^{k}c(i)\le \frac{B}{2}.$
We get
$$\frac{v(S_k)-v(S\cup\{j\})}{B/2}\le\frac{v(S_k)-v(S\cup\{j\})}{c(S_k)}\le\frac{v(S_k\cup S)-v(S\cup\{j\})}{b(S_k\cup S)-b(S\cup\{j\})}<\frac{v(S_k)}{B}.$$
Thus, $v(S_k)<2v(S\cup \{j\}).$

Recalling inequality~(2) on $\frac{m_S(j)}{b_j}$, we derive
$$\frac{m_{S_{j-1}}(j)}{b_j}\ge\frac{m_S(j)}{b_j}\ge\frac{2 v(S\cup\{j\})}{B}>\frac{v(S_k)}{B}.$$
Hence, we arrive at the contradiction with $b_j > m_{S_{j-1}}(j)\frac{B}{v(S_k)}$.
\end{proof}

\subsubsection{Approximation Ratio}

Before analyzing the performance of our mechanism, we consider a simple greedy algorithm (without considering bidding strategies): order items according to their marginal contributions and add as many items as possible (i.e., it stops when we cannot add the next item as the sum of $c_i$ otherwise will be bigger than $B$).
Moreover we can consider the fractional variant of that, i.e., for the remaining budget we take a portion of the item at which we have stopped.
Let $\ell$ be the maximal index for which $\sum_{i=1,\ldots,\ell}c_i\le B$. Let $c'_{\ell+1}=B-\sum_{i=1,\ldots,\ell}c_i$ and $m'_{\ell+1}=m_{\ell+1}\cdot \frac{c'_{\ell+1}}{c_{\ell+1}}$. Hence, the fractional greedy solution is defined to be
\[fgre(A)\triangleq \sum_{i=1}^{\ell}m_i + m'_{\ell+1}.\]

It is well-known that the greedy algorithm is a $1-1/e$ approximation to maximization of monotone submodular functions with a cardinality constraint~\cite{NW81}. Also there was shown that the simple greedy algorithm has unbounded approximation ratio in case of weighted items with a capacity constraint. Nevertheless, a variant of greedy was suggested in \cite{KG05} which gives the same $1-1/e$ approximation to the weighted case. Next we present the following lemma, which is fundamental to our analysis, establishing the same approximation ratio for the simple greedy algorithm with fractional solution. (The proof is deferred to Appendix~\ref{appendix-submodular-reduction}.)
%The proof is by a reduction from weighted problem to unweighted problem.

\begin{lemma}\label{submodular reduction}
Fractional greedy solution has approximation ratio $1-1/e$ for the weighted submodular maximization problem. That is,
\[fgre(A)\ge (1-1/e)\cdot opt(A),\]
where $opt(A)$ is the value of the optimal integral solution for the given instance $A$.
\end{lemma}

Now we are ready to analyze the approximation ratio of \RMmodular. Let $T=\{1,\ldots,k\}$ be the subset returned by $\Gmodular(A,\frac{B}{2})$.
For any $j=k+1,\ldots,\ell$, we have $\frac{c_j}{m_j}\ge \frac{c_{k+1}}{m_{k+1}} > \frac{B}{2 \sum_{i=1}^{k+1}m_i}$, where the last inequality follows from the fact that the greedy strategy stops at item $k+1$. Hence, $c_j>B\cdot \frac{m_j}{2 \sum_{i=1}^{k+1}m_i}$. Same analysis shows that $c'_{\ell+1}>B\cdot \frac{m'_{\ell+1}}{2 \sum_{i=1}^{k+1}m_i}$. Therefore,
$$B\cdot \frac{\sum_{j=k+1}^{\ell}m_j+m'_{\ell+1}}{2\sum_{i=1}^{k+1}m_i}<\sum_{j=k+1}^{\ell}c_j+c'_{\ell+1}\le B.$$
Which implies that $2 \sum_{i=1}^{k+1}m_i > \sum_{j=k+1}^{\ell}m_j+m'_{\ell+1}$ and $m_{k+1}+ 2 \sum_{i=1}^{k}m_i > \sum_{j=k+2}^{\ell}m_j+m'_{\ell+1}$.   Hence,
\[fgre(A)=\sum_{i=1}^{\ell}m_i + m'_{\ell+1} =   \sum_{i=1}^{k+1}m_i  +  \sum_{j=k+2}^{\ell}m_j+m'_{\ell+1} < 3\sum_{i\in S}m_i + 2 m_{k+1} \leq 3\sum_{i\in S}m_i + 2 v(i^*).\]

Together with Lemma~\ref{submodular reduction}, we can bound the optimal solution as
\begin{equation}\label{eqn:opt-upper}
opt(A)\le\frac{e}{e-1} \Big( 3 \Gmodular(A,B/2) + 2 v(i^*) \Big).
\end{equation}
Therefore, the expected value of our randomized mechanism is
$\frac{3}{5} \Gmodular(A,B/2) + \frac{2}{5} v(i^*)\ge \frac{e-1}{5e}opt.$

%
%Now we can write
%
%$$E(v(\calM))=\frac{1}{2}\left(v(\ma)+v(B/2-\greedy_B)\right)\ge$$
%$$\frac{1}{2}v(B/2-\fra-\greedy_B)\ge\frac{1}{4}v(\fra-\greedy_{B/2})\ge$$
%$$\frac{1}{8}v(\fra-\opt_{B/2})\ge\frac{1}{16}v(\fra-\opt_{B})\ge v(\opt_{B})$$

\subsection{Derandomization}

In this section, we provide a deterministic truthful mechanism which is budget feasible and has constant approximation ratio (where $opt(A\setminus \{i^*\},B)$ denotes the value of the optimal solution for the weighted submodular maximization problem given instance $A\setminus \{i^*\}$ with budget $B$).

\begin{center}
\small{}\tt{} \fbox{
\parbox{4.0in}{
\hspace{0.05in} \\[-0.05in] \Mmodular
\begin{enumerate}
\item Let $A=\{i~|~c_i\le B\}$ and $i^*\in \arg\max_{i\in A} v(i)$
\item If $\frac{1 + 4 e + \sqrt{1 + 24 e^2}}{2 (e-1)}\cdot v(i^*) \ge opt(A\setminus \{i^*\},B)$,\footnote{} return $i^*$
\item Otherwise, return $\Gmodular(A,B/2)$
\end{enumerate}
}}
\end{center}

\footnotetext[3]{Our deterministic mechanism works in general not in polynomial time because of the hardness of computing an optimal solution for submodular maximization problems. However, we may substitute it by the optimum of the fractional problem; therefore for special problems like knapsack (discussed in the following subsection), we can get a polynomial time deterministic mechanism.
Note however that we cannot replace it by the simple greedy solution as it breaks monotonicity.
%Moreover, any monotone approximation with constant ratio also works here.
Our mechanism suggests a natural question on the power of computation in (budget feasible) mechanism design at the price of being truthful~\cite{PSS08,Dob}. In particular, can an (exponential runtime) mechanism beat the lower bound of all polynomial time mechanisms? We leave this as future work.}

\begin{theorem}
\Mmodular \ is a  budget feasible truthful mechanism for submodular functions with approximation ratio $\frac{6 e -1 + \sqrt{1 + 24 e^2}}{2 (e-1)} (\approx 8.34)$.
\end{theorem}
\begin{proof}
Note that the bid of $i^*$ is independent to the value of $opt(A\setminus \{i^*\},B)$. Therefore, the mechanism is truthful (a detailed similar argument is given in the proof of Theorem~\ref{th_M_knapsack} in Appendix~\ref{appendix-knapsack}). Budget feasibility follows from Lemma~\ref{lemma-pay-upper} and the observation that Step~2 only gives additional upper bounds on the thresholds of winners from $\Gmodular(A,B/2)$.

It remains to prove the approximation ratio. Let $x=\frac{1 + 4 e + \sqrt{1 + 24 e^2}}{2 (e-1)}(\approx 7.34)$.
We observe that $opt(A,B) - v(i^*) \le opt(A\setminus\{i^*\},B) \le opt(A,B).$

If the condition in Step~2 holds and the mechanism outputs $i^*$, then
\[opt(A,B) \le opt(A\setminus\{i^*\},B)+v(i^*) \le (x+1)\cdot v(i^*).\]
Otherwise, the condition in Step~2 fails and the mechanism outputs $\Gmodular(A,B/2)$ in Step~3.
Recall that in formula~(\ref{eqn:opt-upper}),
$opt(A,B)\leq \frac{e}{e-1} \Big( 3 \Gmodular(A,B/2) + 2 v(i^*) \Big).$
We have
\[ x \cdot v(i^*) < opt(A\setminus \{i^*\},B)\le opt(A,B) \leq \frac{e}{e-1} \Big( 3 \Gmodular(A,B/2) + 2 v(i^*) \Big). \]
This implies that
$ v(i^*) \leq \frac{3 e}{x(e-1)-2e} \Gmodular(A,B/2).$
Hence,
\[ opt \leq \frac{e}{e-1} \Big( 3 \Gmodular(A,B/2) + 2 v(i^*) \Big) \le \frac{e}{e-1} \left( 3  + \frac{6 e}{x(e-1)-2e} \right)\cdot \Gmodular(A,B/2).\]

Simple calculations show that
$1+x= \frac{6 e -1 + \sqrt{1 + 24 e^2}}{2 (e-1)}=\frac{e}{e-1} \left( 3  + \frac{6 e}{x(e-1)-2e} \right).$
Therefore, we have $opt \le (x+1)\cdot \Gmodular(A,B/2)$ in the both cases, which concludes the  proof of the claim with approximation ratio
$\frac{e}{e-1} \left( 3  + \frac{6 e}{x(e-1)-2e} \right) (\approx 8.34)$.
\end{proof}

\subsection{Improved Mechanisms for Knapsack}

In this subsection, we consider a special model of submodular functions where the valuations of agents are additive, i.e., $v(S)=\sum_{i\in S}v_i$ for $S\subseteq [n]$. This leads to an instance of the Knapsack problem, where items correspond to agents and the size of the knapsack corresponds to budget $B$. Singer~\cite{PS10} give a 6-approximation deterministic mechanism. By applying approaches from the previous subsections, we have the following results (proofs are deferred to Appendix~\ref{appendix-knapsack}).

\begin{theorem}\label{theorem-knapsack-mechanism}
There are $2+\sqrt{2}$ approximation deterministic and 3 approximation randomized polynomial truthful budget feasible mechanisms for knapsack.
\end{theorem}

\section{Lower Bounds}

In this section we focus on lower bounds for the approximation ratio of truthful budget feasible mechanisms for knapsack. Note that the same lower bounds can be applied to the general submodular functions as well. In~\cite{PS10}, a lower bound of $2$ is obtained by the following argument: Consider the case with two items, both of unit value
(the value of two items together is 2). If their costs are $(B-\epsilon, B-\epsilon)$, at least one item should win,
otherwise the approximation ratio is infinite. Without loss of generality, we can assume that the first item wins, then its payment is
at least $B-\epsilon$. Now consider another profile $(\epsilon, B-\epsilon)$, the first item should also win and get payment at least
 $B-\epsilon$ by truthfulness. Then the second item could not win because of the budget constraint and individual rationality. Therefore, the mechanism
can only archieve value $1$ for that instance while the optimal solution is $2$. This gives us the lower bound of $2$.

We improve the deterministic lower bound to $1+\sqrt{2}$ by a more involved proof. We also adduce a lower bound of $2$ for
universally randomized truthful mechanisms. All our lower bounds are unconditional, which implies that we do not impose any complexity assumption and constraints of the running time on the mechanism. Our lower bounds relys only on truthfulness and budget feasibility.

%All our lower bounds still hold for a special submodular problem called knapsack.
%For knapsack problems, every item has a value and the valuation of a set is simply the summation of the values of all its containing items.
%In the appendix, we have improved mechanism with ratio $2+\sqrt{2}$ (and $3$ for randomized version) for this special problem.

\subsection{Deterministic Lower Bound}

\begin{theorem}\label{theorem-knapsack-lower}
There is no truthful budget feasible mechanism that
can achieve an approximation ratio better than $1+\sqrt{2}$, even if
there are only three items.
\end{theorem}

Assume otherwise that there is a budget feasible truthful mechanism
that can achieve a ratio better than $1+\sqrt{2}$. We consider the
following scenario: budget $B=1$, and values $v_1=\sqrt{2}$,
$v_2=v_3=1$. Then the mechanism on a scenario has the following two
properties: (i) if all items are winners in the optimal solution,
the mechanism must output at least two items; and (ii) if $\{1,2\}$ or $\{1,3\}$ is the optimal solution, the mechanism cannot output either $\{2\}$ or $\{3\}$ (i.e., a single item with
unit value).
For any item $i$, let function $p_i(c_j,c_k)$ be the payment offered
to item $i$ given that the bids of the other two items are $c_j$ and
$c_k$. That is, $p_i(c_j,c_k)$ is the threshold bid of $i$ to be a winner.

\begin{lemma}\label{lemma-knapsack-lower}
For any $c_3>0.5$ and any domain $(a,b)\subset (0,1-c_3)$, there is
$c_2\in (a,b)$ such that $p_1(c_2,c_3) < 1-c_2$.
\end{lemma}
\begin{proof}
Assume otherwise that there are $c_3>0.5$ and domain $(a,b)\subset
(0,1-c_3)$ such that for any $c_2\in (a,b)$, $p_1(c_2,c_3) \ge
1-c_2$. Let $c_1=1-c_3-b$, then $c_1+c_2+c_3<1=B$, which implies
that the mechanism has to output at least two items. Since
$c_1=1-c_3-b<1-c_2\le p_1(c_2,c_3)$, item 1 is a winner. Further,
$p_1(c_2,c_3) \ge 1-c_2 > 0.5$, which together with budget feasibility imply that item 3 cannot be a winner. Therefore, item 2 must be a
winner with payment $p_2(c_1,c_3) = c_2$ due to individual
rationality and budget feasibility. The same analysis still holds if
the true cost of item 2 becomes $c'_2=\frac{c_2+b}{2}$, i.e., item 2
is still a winner with payment $c'_2$. Thus for the sample $(c_1,c_2,c_3)$ the payment $p_2(c_1,c_3) \ge c'_2 > c_2$, a contradiction.
\end{proof}

Since item 2 and 3 are identical, the above lemma still holds if we
switch item 2 and 3 in the claim. We are now ready to prove
Theorem~\ref{theorem-knapsack-lower}.

\bigskip \noindent {\em Proof of
Theorem~\ref{theorem-knapsack-lower}.} Define $c_3=0.7$ and
$(a,b)=(0.2,0.3)$. Note that $c_3$ and $(a,b)$ satisfy the condition
of Lemma~\ref{lemma-knapsack-lower}. Hence, there is $c\in
(0.2,0.3)$ such that $p_1(c,0.7)<1-c$. Define
$p_1(c,0.7)= 1-c-x, \ \textup{where} \ x>0.$
Symmetrically, define $c_2=0.7$ and $(a',b')=(c,\min\{0.3,c+x\})$.
Again by Lemma~\ref{lemma-knapsack-lower}, there is $d\in (a',b')$
such that $p_1(0.7,d)<1-d$. Define
$p_1(0.7,d)=1-d-y, \ \textup{where} \ y>0.$
Pick $c_1=1-d-\epsilon$, where $\epsilon>0$ is sufficiently small so
that $c_1\in (1-c-x,1-c)\cap (1-d-y,1-d)$. Note that since $d\in
(c,c+x)$, $c_1$ is well-defined.

Consider a true cost vector $(c_1,c,0.7)$. Since
$p_1(c,0.7)=1-c-x<c_1$, item 1 cannot be a winner. Since
$c_1+c=1-d-\epsilon+c<1$, the optimal solution has value at least
$v_1+v_2=1+\sqrt{2}$; therefore the mechanism has to output both
item 2 and 3. Hence, $p_3(c_1,c)\ge c_3 = 0.7.$

Similarly, consider true cost vector $(c_1,0.7,d)$; we have
$p_2(c_1,d)\ge c_2=0.7.$
Finally, consider cost vector $(c_1,c,d)$. By the above two
inequalities, both items 2 and 3 are the winners; this contradicts the
budget feasibility. \hfill $\square$

\subsection{Randomized Lower Bound}

\begin{theorem}\label{theorem-knapsack-lower-random}
There is no randomized (universally) truthful budget feasible mechanism that
can achieve an approximation ratio better than $2$, even in case of
two items.
\end{theorem}
\begin{proof}
We use Yao's min-max principle, which is a typical tool used to prove lower bounds. By the principle, we need to design
a distribution of instances and argue that any deterministic budget feasible mechanism cannot get an expected approximation ratio which is better than
$2$.

All the instances contain two items both with value $1$. Their costs $(c_1,c_2)$ are drawn from the following distribution (see Fig.~\ref{figure:square} in Appendix):
\begin{enumerate}
  \item $(\frac{kB}{n}, \frac{(n-k)B}{n})$ with probability $\frac{1-\epsilon}{n-1}$, where $k=1,2,\ldots, n-1$,
  \item $(\frac{iB}{n}, \frac{jB}{n})$ with probability $\frac{2\epsilon}{(n-1)(n-2)}$, where $i,j\in\{1,\ldots,n-1\}$ and $i+j>n$,
\end{enumerate}
where $1>\epsilon >0$ and $n$ is a larger integer.

We first claim that for any deterministic truthful budget feasible mechanism with finite expected approximation ratio, there is at most one instance, for which both items win in the mechanism.
Assume for contradiction that there are at least two such instances.
Note that for the second distribution $(\frac{iB}{n}, \frac{jB}{n})$, where $i+j>n$, it cannot be the case that both items win
due to the budget constraint. Hence, the two instances must be of the first type; denote them as $(\frac{k_1B}{n}, \frac{(n-k_1)B}{n})$
and $(\frac{k_2 B}{n}, \frac{(n-k_2 )B}{n})$, where $k_1 > k_2$. Consider then the instance $(\frac{k_1B}{n}, \frac{(n-k_2 )B}{n})$ .
Since $k_1+ n-k_2 > n$, this is the instance of the second type in our  distribution. Therefore it has nonzero probability (see fig.~\ref{figure:square}). The mechanism has finite approximation ratio, thus it should have finite approximation ratio on the instance $(\frac{k_1B}{n}, \frac{(n-k_2 )B}{n})$ as well. As a result, it cannot be the case that
both items loss. We assume that item 1 wins (the proof for the other case is similar); the payment to him is at least $\frac{k_1B}{n}$ due to
individual rationality. Then consider the original instance $(\frac{k_2 B}{n}, \frac{(n-k_2 )B}{n})$; item 1 should also win and get a threshold payment, which is equal to or greater than $\frac{k_1B}{n}$. Therefore the payment to second item because of the budget constraint is at most $B-\frac{k_1B}{n}= \frac{(n-k_1 )B}{n}$.
Since $\frac{(n-k_1)B}{n}<\frac{(n-k_2)B}{n}$, we arrive at a  contradiction with either individual rationality or assumption that both items won in the instance $(\frac{k_2 B}{n}, \frac{(n-k_2)B}{n})$.

On the other hand, for all instances $(\frac{kB}{n}, \frac{(n-k)B}{n})$, both items win in the optimal solution with value 2.
Hence, the expected approximation ratio of any deterministic truthful budget feasible mechanism is at least
$\frac{1-\epsilon}{n-1}\cdot 1 +  (n-2) \cdot \frac{1-\epsilon}{n-1}\cdot 2+ \epsilon \cdot 1=2-\epsilon -\frac{1-\epsilon}{n-1}.$
The ratio approaches to $2$ when $\epsilon \rightarrow 0$ and $n\rightarrow \infty$. This completes the proof.
\end{proof}

\section{Beyond Submodularity}

A natural generalization of knapsack is to consider heterogeneous items.
That is, we are given $m$ different types of items and each item has a (private) cost $c_i$ and a (public) value $v_i$,
as well as an indicator $t_i\in [m]$ standing for the type of item $i$. The goal is to pick items of different types
(i.e., one cannot pick more than one item of the same type) to maximize total value given a budget constraint $B$.
The knapsack problem studied in the last section is therefore a special case of the heterogeneous problem when all items are of different types.
However, we cannot simply apply the mechanisms for knapsack here because of heterogeneous items. (Notice however that the lower bounds established in the last section still work.)

The main difference of this problem with knapsack or general submodular functions is that here not every subset is a feasible solution. A straightforward greedy could end up with a very poor solution: Consider a situation that every type contains one very small item (both $v_i$ and $c_i$ are very small) but with large value cost ratio $\frac{v_i}{c_i}$; greedy will take all these small items first and therefore not be able to take more since each type already has one item. The overall value of this greedy solution can be arbitrarily bad compared to the optimal solution.

To construct a truthful mechanism for heterogeneous knapsack, we employ a greedy strategy with {\em deletions}. The main idea is that at every time making a greedy move, we consider two possible changes: (i) add a new item whose type has not been considered before, and (ii) replace an existing item with the new one of same type. Among all the possible choices (of two types), we greedily select items with highest value cost ratio: In the case of adding a new item, its value cost ratio is defined as usually $\frac{v_i}{c_i}$. For the replacement case where we replace $i$ with $j$, its marginal value is $v_j-v_i$ and marginal cost is $c_j-c_i$, and hence its value cost ratio is defined to be $\frac{v_j-v_i}{c_j-c_i}$.

As before, now we assume that all the items are ordered according to their appearances in the greedy algorithm (note that some items never appear in the algorithm and we simply ignore them). The following greedy strategy is similar to what we did for the knapsack problem. In Appendix~\ref{appendix-heterogeneous}, we prove that it is monotone (therefore truthful) and budget feasible.
(Here for notation simplicity, we assume that we already take an item with $c=0$ and $v=0$ for each type, thus every greedy step can be viewed as a replacement.)

\begin{center}
\small{}\tt{} \fbox{
\parbox{4.5in}{
\hspace{0.05in} \\[-0.05in] $\greH$
\begin{enumerate}
\item Let $k=1$, $S=\emptyset$, and $last[j]=0$ for $j\in[m]$
\item While $k\le |A|$ and $c(k)-c(last[t_k])\le B\cdot \frac{v(k)-v(last[t_k])}{v(k)-v(last[t_k])+\sum_{i\in S}v(i)}$
			\begin{itemize}
			\item let $S \leftarrow (S\setminus\{last[t_k]\})\cup \{k\}$
			\item let $last[t_k] = k$
			\item let $k\leftarrow k+1$
			\end{itemize}
\item Return winning set $S$
\end{enumerate}
}}
\end{center}

By applying the above $\greH$, we have the following claim for heterogeneous knapsack. (Details can be found in Appendix~\ref{appendix-heterogeneous}.)

\begin{theorem}\label{theorem-star-knapsack-mechanism}
There are $2+\sqrt{2}$ approximation deterministic and 3 approximation randomized polynomial truthful budget feasible mechanisms for knapsack with heterogeneous items.
\end{theorem}

Finally, we comment that greedy is typically the first choice when we consider designing truthful mechanisms because it usually has a nice monotone property. However, when we allow cancelations in the greedy process, its monotonicity may fail. In the heterogeneous knapsack problem, fortunately $\greH$ is still monotone (although its proof is much more involved) and therefore we are able to apply it to design truthful mechanisms with good approximation ratios. Our idea sheds light on the possibility of exploring budget feasible mechanisms in larger domains beyond submodularity.

\newpage
\appendix
\input{appendix.tex}

\end{document}

%% file: appendix.tex
\begin{figure}
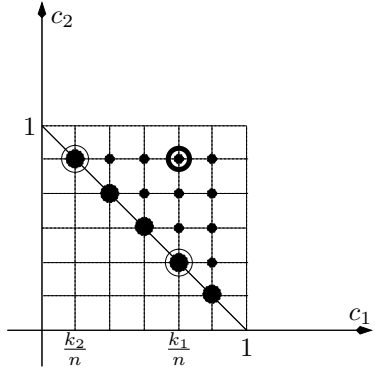

\begin{center}
\include{square}
\end{center}
\caption{distribution for $n=6$}
\label{figure:square}
\end{figure}

\section{Proof of Lemma \ref{submodular reduction}}\label{appendix-submodular-reduction}

\begin{proof}
Assume that all weights $w_i$ are integers. We reduce our weighted problem with monotone submodular utility function $u$ to the unweighted one as follows.

\begin{itemize}
\item For each item $i\in[n]$ we consider new $w_i$ items of unit weight.
      Denote them as $i_j$ for $j\in[w_i]$ and call $i$ to be the type of the
      unit $i_j$.
\item We make the new valuation function $\nu$ depends only on the amounts of
      unit items of each type.
\item Let a set $S$ contains $a_i$ units of each type $i$. Independently
      for each type pick at random in the set $\calR$ with probability
      $\frac{a_i}{w_i}$ weighted item $i$. Define $\nu(S)=E(u(\calR))$.
\end{itemize}

Therefore
$$\nu(S)=\frac{1}{w_1\cdot\ldots\cdot w_n}\sum_{\pi}u\left(S\cdot \pi\right)$$
where $\pi$ is a sampling of units one for each type (there are $w_1\cdot\ldots\cdot w_n$ variants for $\pi$); $S\cdot \pi$ is a vector of types at which $\pi$ hits $S$.

Using this formula it is not hard to verify monotonicity and submodularity of $\nu$. Indeed, e.g. to verify submodularity one only need to check that the marginal contribution of any unit is smaller for a large set, i.e. for $S\subset T$ and $i_j\notin T$ verify inequality $\nu(S\cap\{i_j\})-\nu(S)\ge\nu(T\cap\{i_j\})-\nu(T)$, which is pretty straightforward.

For any $T\subseteq [n]$ if we consider a set of units $S=\{i_k|i\in T, 1\le k\le w_i\}$, then according to the definition $\nu(S)=u(T)$. Hence, optimal solution to the unit weights problem is large or equal than the optimal solution to the original problem.

To conclude the proof it is only left to show that the greedy scheme for the unit weights gives us the same result as our fractional greedy scheme for an integer weights. Note that once we have taken a unit of type $i$ we will proceed to take units of type $i$ until exhaust it completely (we brake ties in favor to the last type we have picked). Indeed, let $i_k, i_{k+1}\notin S$ then
\begin{eqnarray*}
\nu(S\cup\{i_k\})-\nu(S) &=& \nu(S\cup\{i_{k+1}\})-\nu(S) \\
&=& \frac{1}{w_1\cdot\ldots\cdot w_n}\sum_{\{\pi|i_{k+1}\in\pi\}} u\left(S\cup\{i_{k+1}\}\cdot\pi\right)-u\left(S\cdot\pi\right) \\
&=& \frac{1}{w_1\cdot\ldots\cdot w_n}\sum_{\{\pi|i_{k+1}\in\pi\}} u\left(S\cup\{i_k,i_{k+1}\}\cdot\pi\right)-u\left(S\cup \{i_k\}\cdot\pi\right) \\
&=& \nu(S\cup\{i_k,i_{k+1}\})-\nu(S\cup\{i_k\})
\end{eqnarray*}

Therefore, marginal contribution of the type $i$ does not decrease if we include in the solution units of type $i$. On the other hand, because $\nu$ is submodular, marginal contribution of any other type can not increase. So we will take unit $i_{k+1}$ right after $i_k$.

Assume we already have picked set $S$ and now are picking the first unit of a type $i$. Hence, $S$ comprises all units of a type set $T$. Then we have
\[\nu\left(S\cup\{i_1\}\right)-\nu\left(S\right)=
\frac{1}{\prod_{k=1}^{n}w_k}\sum_{\{\pi|i_{1}\in\pi\}} u\left(S\cup\{i_1\}\cdot\pi\right)-u\left(S\cdot\pi\right)
=\frac{\prod_{k\neq i}w_k}{\prod_{k=1}^{n}w_k}m_{T}(i)=\frac{m_T(i)}{w_i}\]
Thus $i=argmax_{i\notin T}\frac{m_T(i)}{w_i}$ which coincides with the rule of our fractional greedy scheme.

In case of real weights the same approach can be applied but in more tedious way.
\end{proof}

\section{Mechanisms for Knapsack}\label{appendix-knapsack}

In this section, we describe our deterministic and randomized mechanisms for knapsack, yielding a proof for Theorem~\ref{theorem-knapsack-mechanism}.

\subsection{Deterministic Mechanism}

We consider the following greedy strategy studied by Singer~\cite{PS10}.

\begin{center}
\small{}\tt{} \fbox{
\parbox{4.0in}{
\hspace{0.05in} \\[-0.05in] $\gre(A)$
\begin{enumerate}
\item Order all items in $A$ s.t. $\frac{v_1}{c_1}\ge \frac{v_2}{c_2}\ge \cdots \ge \frac{v_{|A|}}{c_{|A|}}$
\item Let $k=1$ and $S=\emptyset$
\item While $k\le |A|$ and $c_k\le B\cdot \frac{v_k}{\sum_{i\in S\cup\{k\}}v_i}$
\begin{itemize}
\item $S \leftarrow S\cup \{k\}$
\item $k\leftarrow k+1$
\end{itemize}
\item Return winning set $S$
\end{enumerate}
}}
\end{center}

It is shown that the above greedy strategy is monotone (and therefore truthful).
Actually, it has the following remarkable property: any $i\in S$ cannot control the output set given that $i$
is guaranteed to be a winner. That is, if the winning sets are $S$ and $S'$ when $i$ bids $c_i$ and $c'_i$,
respectively, where $i\in S\cap S'$, then $S=S'$. Otherwise, consider the item $i_0\notin S\cap S'$ with the smallest index;
assume without loss of generality that $i_0\in S\setminus S'$. Let $T = \{j\in S\cap S'~|~j<i_0, j\neq i\}$ be the winning
items in $S\cap S'\setminus \{i\}$ before $i_0$. Then
$c_{i_0}\le B\cdot \frac{v_{i_0}}{\sum_{j\in S}v_j} \le B\cdot \frac{v_{i_0}}{\sum_{j\in T}v_j + v_i+v_{i_0}}$, which implies that
$i_0$ should be a winner in $S'$ as well, a contradiction.

Given the greedy strategy described above, our mechanism for
knapsack is as follows (where $fopt(A)$ denotes the value of the 
optimal fractional solution; for knapsack it can be computed in polynomial time).

\begin{center}
\small{}\tt{} \fbox{
\parbox{4.0in}{
\hspace{0.05in} \\[-0.05in] \Mknapsack
\begin{enumerate}
\item Let $A=\{i~|~c_i\le B\}$ and $i^*\in \arg\max_{i\in A} v_i$
\item If $(1+\sqrt{2})\cdot v_{i^*} \ge fopt(A\setminus \{i^*\})$, return $i^*$
\item Otherwise, return $S= \gre(A)$
\end{enumerate}
}}
\end{center}

\begin{theorem}\label{th_M_knapsack}
\Mknapsack\ is a $2+\sqrt{2}$ approximation budget feasible truthful
mechanism for knapsack.
\end{theorem}
\begin{proof}
The proof consists of each property stated in the claim.
\begin{itemize}
\item {\em Truthfulness.} We analyze monotonicity of the mechanism according to the condition of Step~2 and 3, respectively. If $i^*$ wins in Step~2 (note that the fractional optimal value computed in Step~2 is independent of the bid of $i^*$), then $i^*$ still wins if he decreases his bid.

    If the condition in Step~2 fails and the mechanism runs Step~3, for any $i\in S$, the subset $S$ remains the same if $i$ decreases his bid. Note that if $i\neq i^*$, when $i$ decreases his bid, the value of the fractional optimal solution computed in Step~2 will not decrease. Hence $i$ is still a winner, which implies that the mechanism is monotone.

\item {\em Individual rationality and budget feasibility.} If $i^*$ wins in Step~2, his payment is the threshold bid $B$. Otherwise, assume that all buyers in $A$ are ordered by $1,2,\ldots,n$; let $S=\{1,\ldots,k\}$. Note that it is possible that $i^*\in S$. For any $i\in S$, let $q_i$ be the maximum value that $i$ can bid such that the fractional optimal value on instance $A\setminus \{i^*\}$ is still larger than $v_{i^*}$. Note that $c_i\le q_i$.

    The payment to any winner $i\in S\setminus \{i^*\}$ is
    $p_i=\min\left\{v_i\cdot \frac{c_{k+1}}{v_{k+1}},B\cdot \frac{v_i}{\sum_{j\in S}v_j}, q_i\right\}$,
    and $p_{i^*}=\min\left\{v_{i^*}\cdot \frac{c_{k+1}}{v_{k+1}},B\cdot \frac{v_{i^*}}{\sum_{j\in S}v_j}\right\}$ if $i^*\in S$.
    It can be seen that the mechanism is individually rational. Further,
    $\sum_{i\in S}p_i \le \sum_{i\in S} B\cdot \frac{v_i}{\sum_{j\in S}v_j} = B$, which implies that the mechanism is budget feasible.

\item {\em Approximation.} Assume that all buyers in $A$ are ordered by $1,2,\ldots,n$, and $T=\{1,\ldots,k\}$ is the subset returned by $\gre(A)$. Let $\ell$ be the maximal item for which $\sum_{i=1,\ldots,\ell}c_i\le B$. Let $c'_{\ell+1}=B-\sum_{i=1,\ldots,\ell}c_i$ and $v'_{\ell+1}=v_{\ell+1}\cdot \frac{c'_{\ell+1}}{c_{\ell+1}}$. Hence, the optimal fractional solution is \[fopt(A)=\sum_{i=1}^{\ell}v_i + v'_{\ell+1}\]

    For any $j=k+1,\ldots,\ell$, we have $\frac{c_j}{v_j}\ge \frac{c_{k+1}}{v_{k+1}} > \frac{1}{v_{k+1}}\cdot B\cdot \frac{v_{k+1}}{\sum_{i=1}^{k+1}v_i}$, where the last inequality follows from the fact that the greedy strategy stops at item $k+1$. Hence, $c_j>B\cdot \frac{v_j}{\sum_{i=1}^{k+1}v_i}$. Same analysis shows $c'_{\ell+1}>B\cdot \frac{v'_{\ell+1}}{\sum_{i=1}^{k+1}v_i}$. Therefore,
    $B\cdot \frac{\sum_{j=k+1}^{\ell}v_j+v'_{\ell+1}}{\sum_{i=1}^{k+1}v_i}<\sum_{j=k+1}^{\ell}c_j+c'_{\ell+1}<B$,
    which implies that $\sum_{i=1}^{k}v_i > \sum_{j=k+2}^{\ell}v_j+v'_{\ell+1}$. Hence,
    \[fopt(A)=\sum_{i=1}^{\ell}v_i + v'_{\ell+1} < 2\sum_{i\in S}v_i + v_{i^*}\]

    A basic observation of the mechanism is that
    \[fopt(A)-v_{i^*} \le fopt(A\setminus\{i^*\})\le fopt(A)\]
    Hence, if the condition in Step~2 holds and the mechanism outputs $i^*$, we have
    \[ fopt(A) \le fopt(A\setminus \{i^*\}) + v_{i^*}\le (2+\sqrt{2})\cdot v_{i^*}\]
    If the condition in Step~3 fails and the mechanism outputs $S$ in Step~4, we have
    \[(1+\sqrt{2})\cdot v_{i^*} < fopt(A\setminus\{i^*\})\le fopt(A)  < 2\sum_{i\in S}v_i + v_{i^*}\]
    which implies that $v_{i^*} < \sqrt{2}\cdot \sum_{i\in S}v_i$. Hence,
    \[opt\le fopt(A)=\sum_{i=1,\ldots,\ell}v_i + v'_{\ell+1} < 2\sum_{i\in S}v_i + v_{i^*}\le (2+\sqrt{2})\cdot \sum_{i\in S}v_i.\]
    Therefore, the mechanism is $(2+\sqrt{2})$ approximation.
\end{itemize}

\end{proof}

\subsection{Randomized Mechanism}

Our randomized mechanism for knapsack is as follows.

\begin{center}
\small{}\tt{} \fbox{
\parbox{3.5in}{
\hspace{0.05in} \\[-0.05in] \RMknapsack
\begin{enumerate}
\item Let $A=\{i~|~c_i\le B\}$ and $i^*\in \arg\max_{i\in A} v_i$
\item With probability $\frac{1}{3}$, return $i^*$
\item With probability $\frac{2}{3}$, return $\gre(A)$
\end{enumerate}
}}
\end{center}

\begin{theorem}\label{theorem-random-knapsack}
\label{th_RM_knapsack}
\RMknapsack\ is a $3$ approximation universal truthful budget feasible
mechanism for knapsack.
\end{theorem}
\begin{proof}
Since both mechanisms in Step~2 and 3 are budget feasible and truthful, it is left only to prove approximation ratio.

Using the same notation and argument in the proof of Theorem~\ref{th_M_knapsack}, assume that all buyers in $A$ are ordered by $1,2,\ldots,n$, and $T=\{1,\ldots,k\}$ is the subset returned by $\gre(A)$. Let $\ell$ be the maximal item for which $\sum_{i=1,\ldots,\ell}c_i\le B$. Let $c'_{\ell+1}=B-\sum_{i=1,\ldots,\ell}c_i$ and $v'_{\ell+1}=c'_{\ell+1}\cdot \frac{v_{\ell+1}}{c_{\ell+1}}$. Hence, the optimal fractional solution is \[fopt(A)=\sum_{i=1}^{\ell}v_i + v'_{\ell+1}\]
and
\[fopt(A)=\sum_{i=1}^{\ell}v_i + v'_{\ell+1} < v_{i^*} + 2\sum_{i\in S}v_i.\]
%For any $j=k+1,\ldots,\ell$, we have $\frac{c_j}{v_j}\ge \frac{c_{k+1}}{v_{k+1}} > \frac{1}{v_{k+1}}\cdot B\cdot \frac{v_{k+1}}{\sum_{i=1}^{k+1}v_i}$, where the last inequality follows from the fact that the greedy strategy stops at item $k+1$. Hence, $c_j>B\cdot \frac{v_j}{\sum_{i=1}^{k+1}v_i}$. Same analysis shows $c'_{\ell+1}>B\cdot \frac{v'_{\ell+1}}{\sum_{i=1}^{k+1}v_i}$. Therefore, $B\cdot \frac{\sum_{j=k+1}^{\ell}v_j+v'_{\ell+1}}{\sum_{i=1}^{k+1}v_i}<\sum_{j=k+1}^{\ell}c_j+c'_{\ell+1}<B$, which implies that $\sum_{i=1}^{k}v_i > \sum_{j=k+2}^{\ell}v_j+v'_{\ell+1}$. Hence,
%\[opt=\sum_{i=1}^{\ell}v_i + v'_{\ell+1} < v_{i^*} + 2\sum_{i\in S}v_i.  \]
The excepted value of \RMknapsack\ is therefore
\[\frac{1}{3}v_{i^*} + \frac{2}{3}\sum_{i\in S}v_i = \frac{1}{3}\Big(v_{i^*} + 2\sum_{i\in S}v_i\Big) > \frac{1}{3} opt \]
which completes the proof.
\end{proof}

\section{Knapsack with Heterogeneous Items}\label{appendix-heterogeneous}

In this section we analyze heterogeneous knapsack problem and $\greH$, which leads to a proof of Theorem~\ref{theorem-star-knapsack-mechanism}.

\subsection{Optimal Fractional Solution}

We start our study again on fractional solutions to the optimization problem. First we have to define what is a fractional relaxation for heterogeneous knapsack or more precisely what is a feasible fractional solution.

\begin{defi}
A feasible solution for heterogeneous knapsack is an $n$-tuple of real numbers $(\alpha_1,\ldots,\alpha_n)\in [0,1]^n$ satisfying  $\sum_{i=1}^{n}\alpha_i c_i \le B$ and $\sum_{i\in t^{-}_j}\alpha_i \le 1$ for any $j\in [m]$. An optimal fractional solution is a feasible solution that maximizes $\sum_{i=1}^{n}\alpha_i v_i$.
\end{defi}

We have the following observation on optimal solution.

\begin{lemma}
\label{lemma_HK}
For a given budget $B$ we can pick an optimal fractional solution $f_{OPT}$ such that
\begin{itemize}
\item there are at most two nonzero amounts of items of any type in $f_{OPT}$.  \item there is exactly one item of any
       type in $f_{OPT}$ except maybe only for one type.
\end{itemize}
\end{lemma}
\begin{proof}

Consider any optimal solution $f^{'}_{OPT}$. Fix the price $p_j$ spent on the particular type $j$ in it. We can use only two items of type $j$ in order to provide the maximum value for the price $p_j$. Indeed, if one draws all items of type $j$ in the plain with $x$-coordinate corresponding to the cost and $y$-coordinate corresponding to the value of an item together with the point $(0,0)$, then the condition $\sum_{i\in t^{-}_j}\alpha_i \le 1$ will describe a point in the convex hull of the drawn set.

Thus we can take $f_{OPT}$ with at most two items of a type and derive the first part of the lemma.

One can derive the second part of the lemma by changing $p_{j_1}$ and $p_{j_2}$ in $f_{OPT}$ such that $p_{j_1}+p_{j_2}$ remains constant. Indeed,
appealing to the picture again, we consider two convex polygons $P_1$ and $P_2$ for the types $j_1$ and $j_2$. If both prices $p_{j_1}$ and $p_{j_2}$ get strictly inside the corresponding sides of those polygons, then by stirring $p_{j_1}$ and $p_{j_2}$ in $f_{OPT}$ with keeping $p_{j_1}+p_{j_2}$ constant we can get to a vertex of $P_1$ or $P_2$ and do not decrease the total value.
\end{proof}

The following algorithm computes an optimal fractional solution for heterogeneous knapsack. (For convenience we add an item numbered by $0$ of a new type with cost $0$ and value $0$; this does not affect any optimal solution.)

\begin{center}
\small{}\tt{} \fbox{
\parbox{6.0in}{
\hspace{0.05in} \\[-0.05in] \fHK
\begin{enumerate}
\item For each type $j\in[m]$, (partially) order items of type $j$ as follows:
    \begin{itemize}
	\item let $last = 0$, $tg = 0$ and $A_j=\emptyset$
	\item while $v(last)< \max\limits_{i\in t^{-}_j}v(i)$
		\begin{itemize}
        \item let $k = \arg\max_{i\in t^{-}_j}	\frac{v(i)-v(last)}{|c(i)-c(last)|}$ and add $k$ to $A_j$
		\item define $tg_k = \frac{v(k)-v(last)}{|c(k)-c(last)|}$
		\item let $last = k$
		\end{itemize}
	\end{itemize}
\item Comprise all $A_j$ into one big set $A$ and order all items s.t. $tg_1\ge \cdots\ge tg_{|A|}$
\item Let $last[j]=0$ for each $j\in[m]$, $\alpha_i=0$ for each $i\in[n]$ and $k=1$
\item While $k\le|A|$ and $c_k + \sum_{i=1}^{k-1}\alpha_{i}\cdot c_i\le B$
	\begin{itemize}
	\item let $\alpha_{last[t_k]}\leftarrow 0$
	\item let $last[t_k]\leftarrow k$, $\alpha_k\leftarrow 1$
	\item let $k\leftarrow k+1$
	\end{itemize}
\item If $k\le|A|$, then let $\alpha_k = \frac{B-\sum_{i=1}^{k-1}\alpha_{i}c_i}{c_k}$ and $\alpha_{last[t_k]} = 1-\alpha_k$
\item Return vector $(\alpha_i)_{i\in [n]}$					
\end{enumerate}
}}
\end{center}

\begin{figure}
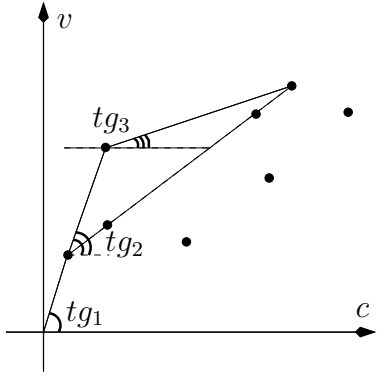

\begin{center}
\include{fig2}
\end{center}
\caption{Convex hull}
\label{figure:star}
\end{figure}

\begin{theorem}
\fHK\ computes an optimal fractional solution for heterogeneous knapsack.
\end{theorem}
\begin{proof}~

\begin{itemize}
\item If we draw every item $i\in t^{-}_j\cup\{0\}$ as a point $(c_i,v_i)$ in
      the plain (see fig.~\ref{figure:star}), then all picked items in $A_j$ will correspond to the part of
      convex hull's vertices of the drawn set from $(0,0)$ to the item with
      maximal value. Computed value of $tg$ will correspond then to the
      tangent of the side of the convex hull with the right end at the given
      item.
\item As in the proof of lemma \ref{lemma_HK} one can find the optimal value,
      that we can get for a type $j$ at the price $c$, by taking the
      $y$-coordinate of the point on a side of convex hull with $c$ at
      $x$-coordinate. Thus for the optimal fractional
      solution we only need items from $A=\cup_j A_j$.
\item Taking everything above into account we can reduce the heterogeneous
      knapsack to the basic knapsack problem. Fix a type $j$ and construct the
      instance of the reduced problem $\tilde{K}_j$ as follows. For each item
      $k\in A_j$ assign the cost $\tilde{c}_k:=c_k-c(last[t_k])$ and the value
      $\tilde{v}_k:=v_k-v(last[t_k])$. It is easy to see that optimal
      solution to basic knapsack problem $\tilde{K}_j$ gives the same value as
      the solution to the original heterogeneous problem restricted to the
      items of type $j$ for any given budget. Hence the optimal fractional
      solution to basic knapsack problem $\cup_j\tilde{K}_j$ has the same
      value as the optimal fractional solution to the original problem.
\item Now it easy to check that our algorithm at stages $2-5$ computes the
      optimal fractional solution to the reduced knapsack problem and
      thus finds the optimal fractional solution to our original problem.
\end{itemize}

\end{proof}

\subsection{Greedy Strategy with Deletions}

We consider the following greedy strategy mechanism.

\begin{center}
\small{}\tt{} \fbox{
\parbox{5.0in}{
\hspace{0.05in} \\[-0.05in] $\greH$
\begin{enumerate}
\item Take the same ordered set $A$ as in Step~2 of \fHK
%    The following rule will do it
%	\begin{itemize}
%	\item Let $last[j]=0$ for $j\in[m]$, $A=\emptyset$
%	\item while $max_{i\in[n]\setminus A} \frac{v(i)-v(last[t_i])}{|c(i)-c(last[t_i])|} > 0$
%        \begin{itemize}
%    	\item let $k\leftarrow argmax_{i\in[n]\setminus A} \frac{v(i)-v(last[t_i])}{|c(i)-c(last[t_i])|}$
%        \item add $k$ to $A$
%		\item let $last[t_k]\leftarrow k$
%		\end{itemize}
%	\end{itemize}
\item Let $k=1$, $S=\emptyset$, and $last[j]=0$ for $j\in[m]$
\item While $k\le |A|$ and $c(k)-c(last[t_k])\le B\cdot \frac{v(k)-v(last[t_k])}{v(k)-v(last[t_k])+\sum_{i\in S}v(i)}$
			\begin{itemize}
			\item let $S \leftarrow (S\setminus\{last[t_k]\})\cup \{k\}$
			\item let $last[t_k] = k$
			\item let $k\leftarrow k+1$
			\end{itemize}
\item Return winning set $S$
\end{enumerate}
}}
\end{center}

Recall the notation in the algorithm \fHK, $tg_k=\frac{v(k)-v(last[t_k])}{|c(k)-c(last[t_k])|}$, where $last[t_k]$ is the last item of type $t_k$ in $A$ at the moment when we are about to add $k$ into $A$. Define $S_k=(S\setminus\{last[t_k])\cup\{k\}$.
Then the second condition in Step~3 of $\greH$ can be rewritten as $$tg_k\ge \frac{v(S_k)}{B}$$

We next analyze the mechanism \greH. Let us denote by $\calM_b$ the run of mechanism $\greH$ on bid $b$ (with the corresponding ordered set $A_b$, the last item of each type $last_b[t_k]$ and marginal tangent $tg_k(\calM_b)$).

\begin{claim}
$\greH$ is monotone (and therefore truthful).
\end{claim}

\begin{proof}
We will show that any losing item cannot bid more and become a winner.
Assume otherwise that item $j$ loses with bid $c_j$ but wins with bid $b_j>c_j$, given that all others bid $c_i$, $i\neq j$.

Note that when $j$ changes his bid, it will only affect the convex hull of items in $t_j^-\cup \{0\}$. The following observations can be verified easily (see fig.~\ref{figure:star}):
\begin{enumerate}
\item Values $v(S)$ of the set of winners and $v(last[t_k])$ for each type $t_k$, taking dynamically in the process of the mechanism, keep increasing.
\item Value $tg_j$ decreases when $j$ increases its bid (since point $(b_j,v_j)$ is on the right hand side of point $(c_j,v_j)$).
\item Ordered set $A_b\setminus t^{-}_j$ is the same as ordered set $A_c\setminus t^{-}_j$
\end{enumerate}

By considering the convex hull for $t^{-}_j$, one can easily see that
if $j$ was not getting at any moment in the winning set $S$ in $\calM_c$ it also will never get in the winning set in $\calM_b$.

Let us explain why for $j$ increasing its bid can not help to remain in the winning set if for the current cost $c_j$ it has been dropped off.

Note that in the new ordered set $A_b$, there can be new items of the same type as $j$ (e.g. $last_{c}[j]$ can be different from $last_{b}[j]$), but nevertheless $tg_j(\calM_b)\le tg_j(\calM_c)$. Let $j'\in t^{-}_j$ be the item that substitutes $j$ in $\calM_c$, then $tg_{j'}(\calM_c)\le tg_{j'}(\calM_b)$ (note that $j'$ necessarily appears in $A_b$). Let $k$ be an item at which $\calM_b$ has stopped, i.e. the first item that we have not taken in the winning set. Assume $k$ stands in $A_b$ not further than $j'$.
Consider two cases.

\begin{enumerate}
\item Let $t_k\neq t_j$. Then
			\begin{itemize}
			\item $tg_{k}(\calM_c)=tg_{k}(\calM_b)$
			\item $v(S_{j'}(\calM_c))\ge v(S_k(\calM_c))$, as $j'$ stands later than $k$ in $A_c$
			\item $v(S_k(\calM_c))=v(S_k(\calM_b))$, since in both $S_k(\calM_b)$
						and $S_k(\calM_c)$ for $t_j$ type we have taken $j$ as well as for
						each other type we have taken the same item.
			\end{itemize}

\item $t_k=t_j$. Then
			\begin{itemize}
			\item $tg_{j'}(\calM_c)\le tg_{j'}(\calM_b)\le tg_{k}(\calM_b)$
			\item $v(S_{j'}(\calM_c))\ge v(S_k(\calM_b))$. The last equality
						holds true, because for each type the value of the item in
						$S_{j'}(\calM_c)$ is greater or equal than value of the
						corresponding item in $S_{k}(\calM_b)$.
			\end{itemize}
		
\end{enumerate}

In both cases we can write
$$tg_k(\calM_b)\ge tg_{j'}(\calM_b)\ge tg_{j'}(\calM_c)\ge\frac{v(S_{j'}(\calM_c))}{B}\ge\frac{v(S_{k}(\calM_b))}{B}$$

Thus we have to take $k$ in $\calM_b$ to the winning set. Hence we arrive at a contradiction. Hence we have taken $j'$ to the winning set in $\calM_b$ and therefore exclude $j$.
\end{proof}

Unfortunately, in contrast to knapsack case this scheme does not possess the following property: any $i\in S$ cannot control the output set given that $i$ is guaranteed to be a winner.

\begin{claim}
\label{bound_payment_hk}
Let $S$ be the winning set of $\greH$ on cost vector $c$. Then no item $j\in S$ can be remained a winner with bid $b_j$ satisfying
$$b_j > (v(j)-v(last_c[t_j]))\cdot\frac{B}{V(S)}+c(last_c[t_j])$$
\end{claim}
\begin{proof}
Assume to the contrary that there exist such $j$ and bid $b_j$. We can write
$$tg_j(\calM_b)=\frac{v(j)-v(last_b[t_j])}{b_j-c(last_b[t_j])}\le \frac{v(j)-v(last_c[t_j])}{b_j-c(last_c[t_j])}<\frac{v(S)}{B}$$

Consider the ordered set $A_c$ and let $k$ be the last item we have taken in the winning set in $\calM_c$. Now consider any item $i\in[1,k]$ where $t_j\neq t_i$. We have $\frac{v(S)}{B}\le tg_k(\calM_c)\le tg_i(\calM_c)=tg_i(\calM_b)$. By the assumption that $j$ is in the winning set in $\calM_b$ and $tg_j(\calM_b)<\frac{v(S)}{B}\le tg_i(\calM_b)$, we get that $S_j(\calM_b)$ contains an item $i'$ with  $t_i=t_{i'}$ and $v(i')\ge v(i)$. Since $j$ is in $S$ and in $S_j(\calM_b)$ we get $v(S_j(\calM_b))\ge v(S)$. Hence

$$\frac{v(S)}{B}> tg_j(\calM_b)\ge\frac{v(S_j(\calM_b))}{B}\ge\frac{v(S)}{B}$$
which gives a contradiction.
\end{proof}

\begin{claim}
\label{budget_feasible_hk}
Greedy scheme $\greH$ is budget feasible.
\end{claim}
\begin{proof}
Let $S$ be a winning set for $\calM$. By Claim~\ref{bound_payment_hk}, we have an upper bound on the payment $p_j$ to each item $j\in S$, i.e.,
$$p_j\le (v(j)-v(last_c[t_j]))\cdot\frac{B}{V(S)}+c(last_c[t_j])$$
Let $0 = i_0, i_1,\ldots ,i_r, i_{r+1} =j$ be the items of type $t_j$ that have appeared in the winning set. We have $tg_{i_\ell}\ge\frac{v(S)}{B}$ for each $\ell=1,\ldots,r$. Hence $$c(i_\ell)-c(i_{\ell-1})\le (v(i_\ell)-v(i_{\ell-1}))\frac{B}{v(S)}$$
Now if we sum up the above inequalities on $c(i_l)-c(i_{l-1})$ for all $\ell=1,\ldots,r$ and plug it in the bound on $p_j$, we get
$$p_j\le\frac{B}{v(S)}\sum_{\ell=1}^{r+1}v(i_\ell)-v(i_{\ell-1})=v(j)\frac{B}{v(S)}$$
Therefore, $\sum_{j\in S}p_j\le B$, which concludes the proof.
\end{proof}

\subsection{Mechanisms}

Given the greedy strategy described above, our mechanism for heterogeneous knapsack is as follows.

\begin{center}
\small{}\tt{} \fbox{
\parbox{4.0in}{
\hspace{0.05in} \\[-0.05in] \MHknapsack
\begin{enumerate}
\item Let $A=\{i~|~c_i\le B\}$ and $i^*\in \arg\max_{i\in A} v_i$
\item If $(1+\sqrt{2})\cdot v_{i^*} \ge \fHK(A\setminus \{i^*\})$, return $i^*$
\item Otherwise, return $S= \greH$
\end{enumerate}
}}
\end{center}

\begin{theorem}
\MHknapsack\ is a $2+\sqrt{2}$ approximation budget feasible truthful
mechanism for heterogeneous knapsack.
\end{theorem}
\begin{proof}
The proof consists of each property stated in the claim.
\begin{itemize}
\item {\em Truthfulness.} The same proof as for knapsack also works here.

\item {\em Individual rationality and budget feasibility.} If $i^*$ wins in Step~2, his payment is the threshold bid $B$. Otherwise, payment to each item has an upper bound from the payment rule in $\greH$ and thus according to the claim \ref{budget_feasible_hk} final total payment will be below given budget $B$.

\item {\em Approximation.} Return back to the algorithm for optimal fractional heterogeneous knapsack. Consider the stage where we add item $k$ to a set $A_j$, let us define $\tilde{v}(k)=v(k)-v(last[t_k])$ and $\tilde{c}(k)=c(k)-c(last[t_k])$ to be modified value and cost of item $k$. Let us consider fractional knapsack $\tilde{FK}$ problem for those modified costs and values for all items in $A$. It turns out that for any budget this new problem $\tilde{FK}$ has the same answer as initial heterogeneous knapsack $HK$. Note that our greedy scheme $\gre$ for modified costs and values and our greedy scheme $\greH$ for original heterogeneous knapsack also give the same answer.
Thus applying the part {\em approximation} of claim \ref{th_M_knapsack} to the modified problem we obtain desired bound.
\end{itemize}
\end{proof}

We can also have the following randomized mechanism with approximation ratio of $3$ (its proof is similar to Theorem~\ref{theorem-random-knapsack}).

\begin{center}
\small{}\tt{} \fbox{
\parbox{3.8in}{
\hspace{0.05in} \\[-0.05in] \RMHknapsack
\begin{enumerate}
\item Let $A=\{i~|~c_i\le B\}$ and $i^*\in \arg\max_{i\in A} v_i$
\item With probability $\frac{1}{3}$, return $i^*$
\item With probability $\frac{2}{3}$, return $S= \greH$
\end{enumerate}
}}
\end{center}

\begin{theorem}
\label{th_RHM_knapsack}
\RMHknapsack\ is a $3$ approximation universal truthful budget feasible
mechanism for heterogeneous knapsack.
\end{theorem}

%% file: square.tex
\begin{picture}(0,0)%
\includegraphics{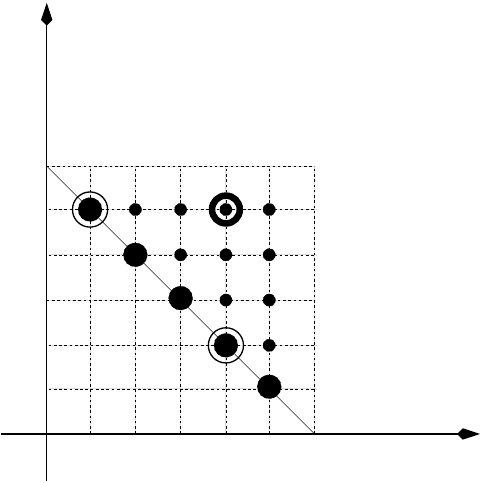}%
\end{picture}%
\setlength{\unitlength}{829sp}%
\begingroup\makeatletter\ifx\SetFigFontNFSS\undefined%
\gdef\SetFigFontNFSS#1#2#3#4#5{%
  \reset@font\fontsize{#1}{#2pt}%
  \fontfamily{#3}\fontseries{#4}\fontshape{#5}%
  \selectfont}%
\fi\endgroup%
\begin{picture}(11001,11001)(3343,-6844)
\put(4681,3449){\makebox(0,0)[lb]{\smash{{\SetFigFontNFSS{10}{12.0}{\rmdefault}{\mddefault}{\updefault}{\color[rgb]{0,0,0}$c_2$}%
}}}}
\put(13591,-5416){\makebox(0,0)[lb]{\smash{{\SetFigFontNFSS{10}{12.0}{\rmdefault}{\mddefault}{\updefault}{\color[rgb]{0,0,0}$c_1$}%
}}}}
\put(10531,-6496){\makebox(0,0)[b]{\smash{{\SetFigFontNFSS{10}{12.0}{\rmdefault}{\mddefault}{\updefault}{\color[rgb]{0,0,0}$1$}%
}}}}
\put(4051,119){\makebox(0,0)[b]{\smash{{\SetFigFontNFSS{10}{12.0}{\rmdefault}{\mddefault}{\updefault}{\color[rgb]{0,0,0}$1$}%
}}}}
\put(5401,-6496){\makebox(0,0)[b]{\smash{{\SetFigFontNFSS{10}{12.0}{\rmdefault}{\mddefault}{\updefault}{\color[rgb]{0,0,0}$\frac{k_2}{n}$}%
}}}}
\put(8506,-6496){\makebox(0,0)[b]{\smash{{\SetFigFontNFSS{10}{12.0}{\rmdefault}{\mddefault}{\updefault}{\color[rgb]{0,0,0}$\frac{k_1}{n}$}%
}}}}
\end{picture}%

%% file: fig2.tex
\begin{picture}(0,0)%
\includegraphics{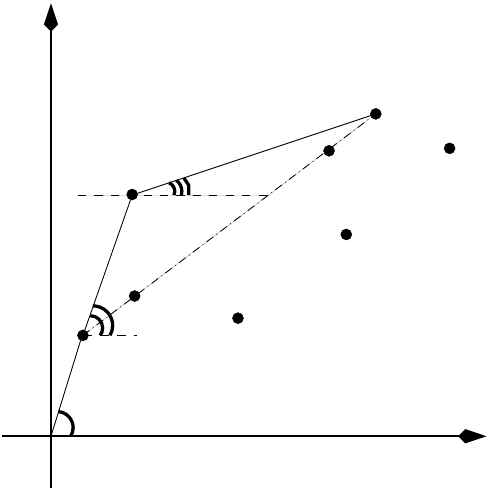}%
\end{picture}%
\setlength{\unitlength}{1036sp}%
\begingroup\makeatletter\ifx\SetFigFontNFSS\undefined%
\gdef\SetFigFontNFSS#1#2#3#4#5{%
  \reset@font\fontsize{#1}{#2pt}%
  \fontfamily{#3}\fontseries{#4}\fontshape{#5}%
  \selectfont}%
\fi\endgroup%
\begin{picture}(8931,8931)(2038,-7069)
\put(3286,1199){\makebox(0,0)[lb]{\smash{{\SetFigFontNFSS{12}{14.4}{\rmdefault}{\mddefault}{\updefault}{\color[rgb]{0,0,0}$v$}%
}}}}
\put(10441,-5731){\makebox(0,0)[lb]{\smash{{\SetFigFontNFSS{12}{14.4}{\rmdefault}{\mddefault}{\updefault}{\color[rgb]{0,0,0}$c$}%
}}}}
\put(3511,-5821){\makebox(0,0)[lb]{\smash{{\SetFigFontNFSS{12}{14.4}{\rmdefault}{\mddefault}{\updefault}{\color[rgb]{0,0,0}$tg_1$}%
}}}}
\put(4456,-4156){\makebox(0,0)[lb]{\smash{{\SetFigFontNFSS{12}{14.4}{\rmdefault}{\mddefault}{\updefault}{\color[rgb]{0,0,0}$tg_2$}%
}}}}
\put(4141,-1186){\makebox(0,0)[lb]{\smash{{\SetFigFontNFSS{12}{14.4}{\rmdefault}{\mddefault}{\updefault}{\color[rgb]{0,0,0}$tg_3$}%
}}}}
\end{picture}%